\RequirePackage{amsmath}
\documentclass[runningheads]{llncs}

\usepackage[T1]{fontenc}
\usepackage{hyperref}
\usepackage{color}

\usepackage{amssymb}

\usepackage{bbm}

\begin{document}
\title{Primal Separation and Approximation\\ for the $\{0, 1/2\}$-closure}
%
%
\author{Lukas Brandl \and
Andreas S. Schulz 
}
\authorrunning{Brandl and Schulz}
%
\institute{Department of Mathematics, School of Computation, Information and Technology and Department of Operations \& Technology, School of Management, Technische Universität München, Germany \email{\{lukas.brandl,andreas.s.schulz\}@tum.de}}
\maketitle              
\begin{abstract}
We advance the theoretical study of $\{0, 1/2\}$-cuts for integer programming problems $\max\{c^T x \colon A x \leq b, x \text{ integer}\}$.  
Such cuts are Gomory-Chv\'atal cuts that only need multipliers of value $0$ or $1/2$ in their derivation. The intersection of all $\{0, 1/2\}$-cuts derived from $Ax \le b$ is denoted by $P_{1/2}$ and called the $\{0,1/2\}$-closure of $P = \{x : Ax \le b\}$. 

The primal separation problem for $\{0, 1/2\}$-cuts is: Given a vertex $\hat x$  of the integer hull of $P$ and some fractional point $x^* \in P$, does there exist a $\{0,1/2\}$-cut that is tight at $\hat x$ and violated by $x^*$? Primal separation is the key ingredient of primal cutting-plane approaches to integer programming. In general, primal separation for $\{0,1/2\}$-cuts is NP-hard. We present two cases for which primal separation 
is solvable in polynomial time. As an interesting side product, we obtain a(nother) simple proof that matching can be solved in polynomial time.  

Furthermore, since optimization over the Gomory-Chv\'atal closure is also NP-hard, there has been recent research on solving the optimization problem over the Gomory-Chv\'atal closure approximately. In a similar spirit, we show that the optimization problem over the $\{0,1/2\}$-closure can be solved in polynomial time up to a factor $(1 + \varepsilon)$, for any fixed $\varepsilon > 0$.



\keywords{Integer Programming \and Gomory-Chv\'atal Cuts \and \{0,1/2\}-Cuts \and Primal Separation \and Approximation}
\end{abstract}
%
%
\section{Introduction}

Let $P = \{x \in \mathbb{R}^n \colon Ax \leq b\}$ be a rational polyhedron with $A \in \mathbb{Z}^{m \times n}$ and $b \in \mathbb{Z}^m$. 
We are interested in solving $\max \{c^T x \colon x \in P \cap \mathbb{Z}^n \}$, which is equivalent to solving $\max \{c^T x \colon x \in P_I\}$. 
Here, $P_I = \text{conv}\{P \cap \mathbb{Z}^n\}$ is the integer hull of $P$. Inequalities of the form $\lambda^T A x \leq \lfloor \lambda^T b \rfloor$, where $\lambda \in \mathbb{R}_+^m$ s.th. $\lambda^T A$ is integer, are valid for $P_I$. Adding them to the description of $P$ strengthens the relaxation. These cuts are called Gomory-Chv\'atal cuts and were first introduced by Gomory \cite{Gomory1958} and Chv\'atal \cite{Chvatal1973}. By applying all valid Gomory-Chv\'atal cuts one obtains $P'$, the so-called Gomory-Chv\'atal closure:
    \begin{align*}
        P' &= \{x \in \mathbb{R}^n \colon \lambda^T A x \leq \lfloor \lambda^T b \rfloor, \lambda \in \mathbb{R}_+^m, \lambda^T A \in \mathbb{Z}^n\} \\
        &= \{x \in \mathbb{R}^n \colon \lambda^T A x \leq \lfloor \lambda^T b \rfloor, \lambda \in [0,1)^m, \lambda^T A \in \mathbb{Z}^n\}.
    \end{align*}
The second equality follows from the well-known fact that all other Gomory-Chv\'atal cuts are dominated by those with multipliers between $0$ and $1$, see, e.g.,~\cite{Conforti2014}.
The Gomory-Chv\'atal closure $P'$ is again a polyhedron \cite{Schrijver1980}, and it is clear that $P_I \subseteq P' \subseteq P.$

Many well-known examples of valid inequalities for combinatorial optimization problems turn out to be Gomory-Chv\'atal cuts with multipliers $\lambda \in \{0, \frac{1}{2}\}^m$, such as the blossom inequalities of the matching polytope \cite{Edmonds1965,Chvatal1973},
the odd-cycle inequalities of the stable set polytope \cite{Gerards1986,Padberg1973}, and the simple comb inequalities for the symmetric traveling salesman problem \cite{Groetschel1979}. 
Furthermore, in the case of the matching polytope, these inequalities 
already generate the integer hull \cite{Edmonds1965}. This motivated Caprara and Fischetti \cite{Caprara1996} to introduce the
$\{0, \frac{1}{2}\}$-closure,
\[
    P_{\frac{1}{2}} = \{x \in P \colon \lambda^T A x \leq \lfloor \lambda^T b \rfloor, \lambda \in \{0, \tfrac{1}{2}\}^m, \lambda^T A \in \mathbb{Z}^n\}.
\]
They showed that, in general, the separation problem (and therefore optimization \cite{Groetschel1981}) for $P_{\frac{1}{2}}$ is $\mathcal{NP}$-hard \cite{Caprara1996} (by a reduction
from the minimum-weight binary clutter problem).
Later, Letchford, Pokutta, and Schulz \cite{Letchford2011} strengthened this result to polytopes contained in the unit cube. Furthermore, Brugger and Schulz~\cite{Brugger2022} showed that deciding whether
a given inequality is valid for $P_{\frac{1}{2}}$ and the separation problem for $P_{\frac{1}{2}}$ are both $\mathcal{NP}$-hard, even if $0 \leq x \leq \mathbbm{1}$ is part of the input.

Eisenbrand \cite{Eisenbrand1999} observed that $P' = P_{\frac{1}{2}}$ holds for the construction by Caprara and Fischetti \cite{Caprara1996}. Therefore, the
separation problem for $P'$ is also $\mathcal{NP}$-hard \cite{Eisenbrand1999}. Cornu\'ejols, Lee, and Li \cite{Cornuejols2018b} showed that this is still true if $P \subseteq [0,1]^n$. They also gave some further hardness results and special cases for which the feasibility problem for $P'$ can be solved in polynomial time. The feasibility problem for $P'$ asks whether $P'$ contains a feasible point (i.e., $P'$ is non-empty).

Koster, Zymolka, and Kutschka \cite{Koster2009}\footnote{An earlier conference version appeared in \cite{ConfKosterZK07}} investigated how to solve the separation problem for $P_{\frac{1}{2}}$ in practice. They introduced preprocessing rules and formulated the remaining problem as an integer linear program. This approach significantly reduced computation time and has since been incorporated into CPLEX \cite{Koster2009}.

It is known that the feasibility problem (and, if $P \subseteq [0,1]^n$, many other problems \cite{Schulz2009}) for $P'$ is in $\mathcal{NP} \cap co\mathcal{NP}$,
if $P' = P_I$ (e.g., \cite{Boyd2009}). This also holds for $P_{\frac{1}{2}}$ if $P_{\frac{1}{2}} = P_I$.
Brugger and Schulz \cite{Brugger2022} showed that recognizing $P_{\frac{1}{2}} = P_I$ is
$\mathcal{NP}$-hard. Previously, Cornu\'ejols and Li \cite{Cornuejols2018}\footnote{An earlier conference version appeared in \cite{ConfCornuejolsL16a}} proved that deciding whether $P' = P_I$ is $\mathcal{NP}$-hard, even if $P_I = \emptyset$. 
Cornu\'ejols, Lee, and Li \cite{Cornuejols2018b} showed that recognizing $P' = P_I$ is still $\mathcal{NP}$-hard if $P \subseteq [0,1]^n$.

In addition to the theoretical significance of the equivalence between separation and optimization \cite{Groetschel1981}, the separation problem assumes a crucial role in both dual fractional cutting plane algorithms, e.g., \cite{Gomory1958}, and branch-and-cut algorithms, see, e.g., \cite{Nemhauser1988}. In the former, a fractional dual solution provided by an LP relaxation is maintained and driven towards integrality by adding cutting planes. The separation problem for a family $\mathcal{F}$ of valid inequalities for $P_I$ plays a key role in these procedures as it provides the cutting planes that are added to the relaxation.

\medskip
    \noindent \textbf{The separation problem for $\mathcal{F}$.} Given some $x^* \in \mathbb{R}^n$, find an inequality in $\mathcal{F}$ that is violated by $x^*$ or decide that none exists.
\medskip


On the other hand, there also exist primal cutting plane algorithms. In these algorithms, an integral feasible solution is maintained at all times. Cutting planes that are tight at this integral solution are added iteratively until a simplex-style pivot to an improving integral neighbor can be performed. For literature on primal cutting plane algorithms, see, e.g., \cite{Padberg1980,Young1968}, or \cite{Letchford2002} for a more recent account.
This approach necessitates a modified version of the separation problem:

\medskip    
\noindent \textbf{The primal separation problem for $\mathcal{F}$.} Given some $x^* \in \mathbb{R}^n$ and a vertex $\hat{x} \in P_I$, find an inequality in $\mathcal{F}$ that is tight at $\hat{x}$ and violated by $x^*$, or decide that none exists.
\medskip

The primal separation problem can always be reduced to the separation problem \cite{PadbergGroetschel1985}. It follows from work by Eisenbrand, Rinaldi, and Ventura \cite{Eisenbrand2003a}\footnote{An earlier conference version appeared in \cite{ConfEisenbrandPrimalSep}} and the polynomial-time equivalence of optimization and separation \cite{Groetschel1981} that the reverse holds for 0/1 polytopes. Furthermore, Letchford and Lodi \cite{letchford2003primal} showed how to reduce the separation problem  to the primal separation problem for certain families of polyhedra. These families include set packing, set covering, knapsack polyhedra and, interestingly, general integer programming. However, their construction does not cover, e.g., traveling salesman polytopes \cite{letchford2003primal}. Using this construction, they show that the primal separation problem for Gomory-Chv\'atal inequalities is $\mathcal{NP}$-hard. It is straightforward to see that their proof also implies the $\mathcal{NP}$-hardness of primal separation for $\{0, \frac{1}{2}\}$-cuts. In addition, they show for multiple families of inequalities that the primal separation problem is \textit{polynomially} easier than their standard separation counterpart.

In their seminal work \cite{Caprara1996}, Caprara and Fischetti 
also showed that the separation problem for $\{0, \frac{1}{2}\}$-cuts can be solved efficiently in the presence of box-constraints if there are at most two odd entries in each column or row of the constraint matrix $A$. 

As the first contribution of this work, we present two procedures that solve the primal separation problem for $\{0, \frac{1}{2}\}$-cuts in polynomial time in these cases. Although the result follows from prior work, specifically \cite{Caprara1996} together with \cite{PadbergGroetschel1985}, our constructions offer a direct, more straightforward, and combinatorial proof. Furthermore, they are conceptually simpler compared to the standard separation methods, as they rely on minimum cut and shortest path computations, in contrast to the standard separation procedures that require \textit{odd} cut and \textit{odd} path computations.

\begin{theorem} \label{thm:primal_sep_col}
    The primal separation problem for $\{0, \frac{1}{2}\}$-cuts can be solved by at most $m + 2n$ min-cut computations 
    if there are at most two odd entries per column of $A$.
\end{theorem}

\begin{theorem} \label{thm:primal_sep_row}
    The primal separation problem for $\{0, \frac{1}{2}\}$-cuts can be solved by at most $m + n$ shortest-path computations 
    if there are at most two odd entries per row of $A$.
\end{theorem}

If, additionally, $P_{\frac{1}{2}} = P_I \subseteq [0,1]^n$ holds in these cases, it follows from \cite{Eisenbrand2003a} (see also \cite{Schulz2009}) that the optimization problem over $P_{\frac{1}{2}} = P_I$ can be solved in polynomial time.

In addition to the result mentioned above, Eisenbrand, Rinaldi, and Ventura \cite{Eisenbrand2003a} provided primal separation procedures for the odd-cut inequalities of the perfect matching polytope, the odd-cycle inequalities of the stable set polytope, the cycle inequalities of the maximum cut polytope, and the odd-cycle inequalities of the maximum bipartite graph polytope that are conceptually simpler than their standard separation counterparts. The same is true  for our procedures when comparing them to the standard separation counterparts given in Caprara and Fischetti \cite{Caprara1996}. 


As a side product, our construction shows that the matching problem can be solved in polynomial time. An analogous result for the perfect matching problem is obtained in \cite{Eisenbrand2003a}, but whereas the method outlined in \cite{Eisenbrand2003a} is specifically tailored for the separation of the odd-cut inequalities of the perfect matching polytope, our construction exhibits a broader applicability. It can be applied to general integer matrices with a maximum of two odd entries per column, general integer right-hand sides, and it is adaptable to any subset of the box-constraints.

In the second part of this work we study approximate optimization over $P_{\frac{1}{2}}$.
While exact optimization over $P'$ is $\mathcal{NP}$-hard, Bienstock and Zuckerberg \cite{Bienstock2005} introduced a hierarchy that gives a $(1 + \varepsilon)$-approximation\footnote{A $(1 + \varepsilon)$-approximation in the context of \cite{Bienstock2005,Fiorini2020,Mastrolilli2020} means a relaxation $Q$ such that $\min\{c^T x \vert x \in P'\} \leq (1 + \varepsilon) \min\{c^Tx \vert x \in Q\}$ for all $c \in \mathbb{Z}^n_+$, and the minimization problem over $Q$ can be solved in polynomial time.} of $P'$ for polytopes of the form $P = \{x \in [0,1]^n \colon Ax \geq \mathbbm{1}\}$,
where $A \in \{0,1\}^{m \times n}.$ Fiorini, Huynh, and Weltge \cite{Fiorini2020} provide a simpler construction based on Boolean formulas that also achieves the same approximation guarantee. 
At the same time, Mastrolilli \cite{Mastrolilli2020}\footnote{An earlier conference version appeared in \cite{ConfMastrolilli}} showed how to obtain the same result by using a modified Sherali-Adams hierarchy involving high-degree polynomials. 
They further showed that the standard Lasserre hierarchy gives a $(1+\varepsilon)$-approximation\footnote{They construct a relaxation $Q$ such that 
$\max\{c^T x \vert x \in Q\} \leq (1 + \varepsilon) \max\{c^Tx \vert x \in P'\}$ for all $c \in \mathbb{Z}^n_+$, and the maximization problem over $Q$ can be solved in polynomial time.} of $P'$ for monotone polytopes~$P$.

The second contribution of this work consists of a relatively simple $(1 + \varepsilon)$-approximation of $P_{\frac{1}{2}}$ if the right-hand side vector $b$ is strictly positive. 
That is, we construct a polyhedron $\Tilde{P}$ such that $P_{\frac{1}{2}} \subseteq \tilde{P} \subseteq (1 + \varepsilon) P_{\frac{1}{2}}$, where $(1 + \varepsilon) P = \{(1 + \varepsilon) x \colon x \in P\}$. Hence, it follows that $\max\{c^T x \colon x \in P_{\frac{1}{2}}\} \leq \max\{c^T x \colon x \in \tilde{P}\}\leq (1 + \varepsilon) \max\{c^T x \colon x \in P_{\frac{1}{2}}\}$. Furthermore, the maximization problem over $\tilde{P}$ can be solved in polynomial time.
Our construction also contains the case of $P$ monotone because we can w.l.o.g. assume that $b >0$ by projecting out variables and in contrast to the result by Mastrolilli \cite{Mastrolilli2020}, $A$ does not have to be non-negative.

%
\section{Primal separation for $\{0, \frac{1}{2}\}$-cuts in the case of at most two odd entries per row or column}
Eisenbrand, Rinaldi, and Ventura \cite{Eisenbrand2003a} showed that primal separation and optimization (and, therefore, separation \cite{Groetschel1981}) are polynomial-time equivalent for 0/1-polytopes.
They also give some examples showing that primal separation can be conceptually much simpler than standard separation. One of their examples is the perfect matching polytope, for which they show that the primal separation problem can be solved by computing $\frac{\vert V \vert}{2}$ maximum flows, where $V$ denotes the set of nodes of the given graph. This stands in contrast to the standard separation problem for the matching polytope, which involves solving a minimum-weight odd cut problem. Minimum-weight odd cut algorithms, such as the first one proposed by Padberg and Rao \cite{Padberg1982}, are based on intricate observations, whereas the algorithm described in \cite{Eisenbrand2003a} is elementary.

Caprara and Fischetti \cite{Caprara1996} reduced the separation problem for the $\{0, \frac{1}{2}\}$-closure to instances of the minimum-weight binary clutter problem. 
Then, if there are at most two odd entries per column (row), they can apply the algorithm by Padberg and Rao \cite{Padberg1982} (Gerards and 
Schrijver \cite{Gerards1986}) to solve the corresponding minimum-weight binary clutter problem in polynomial time.

In the same spirit as \cite{Eisenbrand2003a}, we show that the primal separation problem for $P_{\frac{1}{2}}$ can be solved in polynomial time in those two cases (i.e., at most two odd entries per column or row) by providing two simple constructions based on minimum-cut and shortest-path computations, respectively.

To derive their primal separation procedure, Eisenbrand, Rinaldi, and Ventura \cite{Eisenbrand2003a} exploit that the description of the matching polytope is known \cite{Edmonds1965}. In the case
of the perfect matching polytope, this is equivalent to assuming $P_{\frac{1}{2}} = P_I$. We, however, derive the graphs for our procedures directly from $A$ and therefore obtain a more general construction.

To establish Theorems \ref{thm:primal_sep_col} and \ref{thm:primal_sep_row}, we need the following, widely known lemma, which provides some insight into the structure of non-trivial $\{0, \frac{1}{2}\}$-cuts that are tight at a given integer vertex $\hat{x}$. A $\{0, \frac{1}{2}\}$-cut is called non-trivial if $\lambda^T b \notin \mathbb{Z}$. Furthermore, original inequalities are \textit{involved} in a $\{0, \frac{1}{2}\}$-cut if the corresponding multiplier is non-zero.

\begin{lemma}\label{lemma:psep_cut_structure}
    Given a polytope $P = \{x \in \mathbb{R}^n \colon Ax \leq b\}$, with $A$ and $b$ integer, an integer vertex $\hat{x}$ of $P_{\frac{1}{2}}$, and $\lambda \in \{0, \tfrac{1}{2}\}^m$
    such that $\lambda^T A \in \mathbb{Z}^n$. Then $\lambda^T A x \leq \lfloor \lambda^T b \rfloor$ is a non-trivial $\{0, \frac{1}{2}\}$-cut 
    that is tight at $\hat{x}$ if and only if
    \[
        \lambda^T (b - A \hat{x}) = \frac{1}{2}.
    \]
    In particular, the entries in $\lambda$ are zero for all inequalities that are not tight at $\hat{x}$ except 
    for exactly one that has slack exactly $1$.
\end{lemma}

\begin{proof}
    Let $\hat{s} = b - A \hat{x} \in \mathbb{Z}^m$. Then
    \[
        \lambda^T A \hat{x} = \lambda^T (b - \hat{s}) \in \mathbb{Z}.
    \]
    Suppose that only tight inequalities are involved in $\lambda^T A x \leq \lfloor \lambda^T b \rfloor$, i.e., $\lambda^T \hat{s} = 0$. Then $\lambda^T b \in \mathbb{Z}$, thus the cut is trivial.
    Assuming the slack adds up to more than $1$, i.e., $\lambda^T \hat{s} \geq 1$, then $\lambda^T (b - \hat{s}) < \lfloor \lambda^T b \rfloor$, and the cut is not tight at $\hat{x}$.
    Finally, if $\lambda^T \hat{s} = \frac{1}{2}$, we obtain a non-trivial $\{0, \frac{1}{2}\}$-cut that is tight at $\hat{x}$, because $\lambda^T (b - \hat{s}) = \lfloor \lambda^T b\rfloor$. \qed
\end{proof}

Furthermore, a non-trivial $\{0, \frac{1}{2}\}$-cut is violated by $x^*$, if $\lambda^T A x^* > \lfloor \lambda^T b \rfloor = \lambda^T b - \frac{1}{2}$, i.e. if
\begin{equation}\label{eq:cut_off}
    \lambda^T (b - A x^*) < \frac{1}{2} .
\end{equation}

In the following we are considering polytopes of the form $P = \{x \in \mathbb{R}^n \colon A x \leq b, 0 \leq x \leq \mathbbm{1}\}$.

Let $\lambda \in \{0, \frac{1}{2}\}^m$. If for a coordinate $i \in [n]$, $(\lambda^T A)_i \notin \mathbb{Z}$, the constraints $0 \leq x_i \leq 1$ allow us to \textit{repair} this coordinate, by either rounding up
or down. Note that this might introduce some slack in \eqref{eq:cut_off}. Furthermore, as indicated by Lemma \ref{lemma:psep_cut_structure}, 
we are only allowed to round down (up) if $\hat{x}_i = 0$ ($\hat{x}_i = 1$), if the slack inequality is chosen from $Ax \leq b$.
It follows that the rounding operation is completely determined by $\hat{x}$ and the choice of $\lambda$.
We will denote the corresponding \textit{rounding multipliers} as $\mu_{\text{down}}, \mu_{\text{up}} \in \{0, \frac{1}{2}\}^n$.

Note that we can assume $\mu_{\text{down}}^T \mu_{\text{up}}=0$. If there was a coordinate $i$ with $(\mu_{\text{down}})_i = (\mu_{\text{up}})_i = \frac{1}{2}$, then setting $(\mu_{\text{down}}) = (\mu_{\text{up}}) = 0$ would result in a stronger cut.
Therefore, choosing a box constraint as the slack inequality is equivalent to forcing a rounding operation, i.e., $(\lambda^T A)_i$ should be fractional.
%
%
\subsection{At most two odd entries per column}

In the following, we will identify the non-trivial $\{0, 1\}$-cuts containing $j^*$ with $j^*$-$t$-cuts in a graph $G=(V,E)$, which is constructed as follows. Here and in the following, $j^*$ will denote the slack inequality from Lemma \ref{lemma:psep_cut_structure}.

We will associate each inequality that is tight at $\hat{x}$ and the slack inequality $j^*$ with a node of a graph  (in addition to a special node $t$).

An edge exists between two inequality nodes if they both have an odd entry in the same coordinate. This edge represents that particular coordinate. If it is contained in a $j^*$-$t$-cut in the graph, only one of the two inequalities is involved, and as a result, the corresponding coordinate must be rounded up or down. The slack that is introduced in \eqref{eq:cut_off} is represented by the capacity of this edge.

Furthermore, there will be edges between each node and $t$ representing the slack introduced in \eqref{eq:cut_off} by the selection of this constraint. 

Additionally, an edge connects a node corresponding to an inequality to node $t$ if there is only one tight inequality containing the corresponding coordinate with an odd coefficient. These edges are then assigned their corresponding slack from the rounding operation, as selecting the inequality associated with such a node necessitates rounding the respective coordinate.

Altogether, if we define a $\{0, 1\}$-cut by $\lambda = \frac{1}{2}\chi_U$, where $\chi_U$ denotes the incidence vector of $j^* \in U \subseteq V \setminus t$, and $\mu_{\text{down}}$ and $\mu_{\text{up}}$ as described above, we can identify the non-trivial $\{0, 1\}$-cuts containing $j^*$ with the $j^*$-$t$-cuts in $G$. Furthermore, the capacity $\frac{c(\delta(U))}{2}$ will be equal to the left-hand side in \eqref{eq:cut_off}. More formally

\begin{proof}[of Theorem \ref{thm:primal_sep_col}]
    Let $\mathcal{S} = \{j \in [m] \colon  a_j^T \hat{x} = b_j - 1\}$ and $\mathcal{T} = \{j \in [m] \colon  a_j^T \hat{x} = b_j\}$ and,
    at first, assume that the inequality $j^*$ with slack $1$ from Lemma \ref{lemma:psep_cut_structure} comes from $\mathcal{S}$.
    
    Let $V_1 = \mathcal{T}\, \dot\cup \, \{j^*\}$, $V = V_1 \, \dot\cup \, \{t\}$ be the nodes of the graph.
    The edges representing the \textit{cost} of selecting constraint $v$ are
    \[
        E_1 = \{\{v,t\} \vert \, v \in V_1 \}.
    \]
    The edges between nodes that represent the sharing of an odd coordinate are
    \[
        E_2 = \{\{v,w\}^i \vert \, v,w \in V_1, v \neq w, i\in[n]\colon a_{v,i} \equiv a_{w,i} \equiv 1 \pmod{2}\}.
    \]
    Finally, there might be constraints with an odd entry that do not have an odd \textit{counterpart}. If such a constraint is selected, the corresponding coordinate will need to be rounded. The following edges represent this situation:
    \[
        E_3 = \{\{v,t\}^i \vert \, v \in V_1, i \in [n] \ \forall w \in V_1 \setminus \{v\}\colon a_{v,i} \equiv 1 \land a_{w,i} \equiv 0 \pmod{2}\}
    \]
    The set of edges will be the union of the sets we just defined, that is,
    \[
        E = E_1 \, \dot\cup \, E_2 \, \dot\cup \, E_3
    \]
    Their capacities are given by
    \[
        c(e) = \begin{cases}
            (b_v - a_v^T x^*) &\text{if } e \in E_1 \\
            x_i^* &\text{if } e \in E_2 \cup E_3, \hat{x}_i = 0 \\
            1 - x_i^* &\text{if } e \in E_2 \cup E_3, \hat{x}_i = 1.
        \end{cases}
    \]
    Notice that the cuts $\delta(U)$, where $j^* \in U \subseteq V \setminus t$, are in one-to-one correspondence with the non-trivial $\{0, \frac{1}{2}\}$-cuts that are tight at $\hat{x}$ and contain constraint $j^*$.
    Furthermore, $c(\delta(U)) = \chi_U^T(b - A x^*) + 2\mu_{\text{down}} (0 - (-x^*)) + 2\mu_{\text{up}} (1 - x^*)$, where $\mu_{\text{up}}$ and $\mu_{\text{down}}$ represent rounding up or down, as described above. 
    
    Then there exists a non-trivial $\{0, \frac{1}{2}\}$-cut that is tight at $\hat{x}$, contains $j^*$ and cuts off $x^*$ if and only if there exists a $j^*$-$t$-cut with capacity < $1$.
    
    Assuming that one of the box constraints is the slack inequality, consider the tight inequalities from $\mathcal{T}$ with an odd $i$-th coefficient. Select one of them and proceed as previously described, without considering the potential other inequality (there are at most two).
    This approach is sufficient because simultaneously selecting both the upper and lower box constraints weakens a cut and therefore does not need consideration.

    In total, at most $m + 2n$ minimum capacity cuts need to be determined. \qed
\end{proof}

Actually, the box-constraints are not necessary in the construction above. The possible absence of (some of) them can be treated as follows. If for a coordinate no rounding operation is \textit{available} via a tight box-constraint, the corresponding nodes have to be contracted. Note that this might lead to infeasible sub-problems, but does not affect the correctness of the procedure.

As mentioned earlier, the $\{0, \frac{1}{2}\}$-closure of the the fractional matching polytope is equal to its integer hull \cite{Edmonds1965}. Furthermore, its constraint matrix $A \in \{0,1\}^{\vert V \vert \times \vert E \vert}$ is the incidence matrix of the graph $G = (V, E)$ of the problem, i.e., it has exactly two ones per column. Thus, Theorem \ref{thm:primal_sep_col} is applicable and it follows from \cite{Eisenbrand2003a} that the matching problem can be solved by at most $\vert V \vert + 2 \vert E \vert$ min-cut computations. 

\begin{corollary}
    The matching problem can be solved in polynomial time by computing at most $\vert V \vert + 2 \vert E \vert$ minimum cuts.
\end{corollary}


%
\subsection{At most two odd entries per row}
In the case of at most two odd entries per row, we follow a similar approach as before. This time, the nodes will represent coordinates, and the edges will represent constraints. We will once again only consider constraints that are active at $\hat{x}$ and the slack constraint $j^*$.

That is, we introduce a special node, denoted as $t$, along with a node for each coordinate. A pair of coordinate nodes will be connected by an edge if they both appear as odd entries in the same constraint. The length of these edges is determined by their respective slack values.

Additionally, for each constraint that contains only a single odd entry (including the box constraints), we introduce an edge between the respective coordinate node and the special node $t$. The length of these edges is determined by the corresponding slack value, which, in the case of box constraints, is either $x^*$ or $1 - x^*$.

Since $\lambda A \in \mathbb{Z}^n$ is equivalent to $(2\lambda) A \equiv 0 \pmod{2}$, we require that the sum of the selected inequalities for each coordinate is even. In our context, where coordinates are represented by nodes, this means that we have to select edges in such a way that the degree of each node is even. This implies that we are looking for unions of cycles in the graph $G$. In our construction, every edge is assigned a non-negative length, making it sufficient to search for the shortest cycle containing the edge corresponding to constraint $j^*$. As this edge is always part of such a cycle, the final problem that has to be solved is a shortest path problem between the nodes connected by this edge. Furthermore, the length of such a cycle is equal to the left-hand side in \eqref{eq:cut_off}. More formally

\begin{proof}[of Theorem \ref{thm:primal_sep_row}]
   Let $\mathcal{S} = \{j \in [m] \colon  a_j^T \hat{x} = b_j - 1\}$ and $\mathcal{T} = \{j \in [m] \colon  a_j^T \hat{x} = b_j\}$. The nodes of the graph are given by $V = [n] \, \dot\cup \, \{t\}$ and the edges representing an inequality that has two odd entries connect the involved coordinate nodes, that is,
    \[
        E_1 = \{\{v,w\}^e \vert \, e \in \mathcal{T} \colon a_{e,v} \equiv a_{e,w} \equiv 1 \pmod{2}\}.
    \]
    The constraints from $Ax \leq b$, that have only a single odd entry, connect the respective coordinate node and $t$, i.e.,
    \[
        E_2 = \{\{v,t\}^e \vert \, e \in \mathcal{T}, \forall w \in [n] \setminus \{v\} \colon a_{e,v} \equiv 1 \land a_{e, w} \equiv 0 \pmod{2}\}.
    \]
    Finally, the tight box-constraint will be represented by $E_3$, defined as
    \[
        E_3 = \{\{v,t\} \vert \, v \in [n]\}.
    \]
    Altogether, the edges of the graph are given by
    \[
        E = E_1 \, \dot\cup \, E_2 \, \dot\cup \, E_3.
    \]
    Then, as described above, the edge lengths are defined as
    \[
        l(e) = \begin{cases}
            (b - A x^*)_e &\text{if } e \in E_1 \cup E_2\\
            x_v^* &\text{if } e \in E_3, \hat{x}_v = 0 \\
            1 - x_v^* &\text{if } e \in E_3, \hat{x}_v = 1.
        \end{cases}
    \]
    Note that, similar as before, the cycles containing the edge corresponding to inequality $j^*$ are in one-to-one correspondence to the non-trivial $\{0, \frac{1}{2}\}$-cuts that are tight at $\hat{x}$ and contain constraint $j^*$. Furthermore, $l(C) = \chi_U^T(b - A x^*) + 2\mu_{\text{down}} (0 - (-x^*)) + 2\mu_{\text{up}} (1 - x^*)$ for each such cycle. 
    
    Therefore we are looking for a minimum-length cycle of this type which is equivalent to looking for a shortest $v^*$-$w^*$-path, where $v^*$ and $w^*$ are the end points of the edge belonging to $j^*$.

    It follows that there exists a non-trivial $\{0, \frac{1}{2}\}$-cut that is tight at $\hat{x}$, contains $j^*$ and cuts off $x^*$ if and only the shortest $v^*$-$w^*$-path is of length less than $1 - (b - A x^*)_{j^*}$.

    The case when a box-constraint is selected as the tight inequality works analogously.
    
    In total, at most $m + n$ shortest-paths have to be computed. \qed
\end{proof}

Similar to our previous approach, the inclusion of box constraints is not necessary in this construction. If certain box constraints are missing, we can handle this by omitting the corresponding edges in the graph.
%
\section{An Approximation Scheme for $P_{\frac{1}{2}}$}

In this section, we present an approximation scheme for $P_{\frac{1}{2}}$. 

\begin{definition}\label{def:PTAS}
    Given a polytope $Q = \{x \in \mathbb{R}^n \colon x \geq 0, Ax \leq b\}$ and an objective function vector $c$, where $A$, $b$, and $c$ are integer and $b \geq 0$, 
    a polynomial time approximation scheme (PTAS) for $Q$ is an algorithm that, given a fixed $\varepsilon > 0$, runs in polynomial time in the encoding size of the input and returns a number $\alpha(Q, c)$ such that
    \[
       \max_{x \in Q} c^T x \leq \alpha(Q, c) \leq (1+\varepsilon)  \max_{x \in Q} c^T x.
    \]
\end{definition}
Observe that $0 \in Q$, which is needed for a sound definition of approximation and, e.g., implies that $Q \subseteq (1 + \varepsilon) Q$ for $\varepsilon \geq 0$.


Note that if we can construct a polytope $\tilde{Q}$ such that $Q \subseteq \tilde{Q} \subseteq (1 + \varepsilon) \,Q$ and we can solve the optimization problem over $\tilde{Q}$ in polynomial time, then we obtain a PTAS for $Q$. 

In our setting, it turns out that it suffices to consider $\{0, \frac{1}{2}\}$-cuts such that $\lambda^T \mathbbm{1} \leq k$ for sufficiently large but fixed $k$. This implies that we only need to consider $\{0, \frac{1}{2}\}$-cuts derived from at most $2k$ original inequalities.


The following construction is inspired by the relaxation of the matching polytope that only considers the blossom inequalities of size at most $k$, see, e.g., \cite{Braun2015a,Rothvoss2017a,Sinha2018}\footnote{An earlier conference versions appeared in \cite{ConfBraunP15,ConfRothvossMatching}}.

\begin{theorem}\label{thm:PTAS}
     Let $P = \{x \in \mathbb{R}^n \colon Ax \leq b\}$, with $A$ and $b$ integer. If $b > 0$, there is a PTAS for $P_{\frac{1}{2}}$.
\end{theorem}
\begin{proof}
The proof uses the fact that the contribution to the right-hand side of each inequality that is selected for a $\{0, \frac{1}{2}\}$-cut can be lower-bounded by $\frac{1}{2}$ because $b \geq 1$. Let us consider
    \[
        \tilde{P} := \{x \in P \colon \lambda^T A x \leq \lfloor \lambda^T b \rfloor,  \lambda^T \mathbbm{1} \leq k(\varepsilon), \lambda \in \{0, \tfrac{1}{2}\}^m\}.
    \]
    with $k(\varepsilon) = \lceil 1 + \frac{1}{\varepsilon}\rceil$. That is, $\tilde{P}$ is the intersection of $P$ with all $\{0, \frac{1}{2}\}$-cuts that involve at most $2k(\varepsilon)$ inequalities from $Ax \le b$. 
    Let $\alpha(P, c) := \max\{c^T x \colon x \in \tilde{P}\}$.
    Because $\tilde{P}$ is defined by at most
    $\mathcal{O}(m^{2k})$ inequalities and all coefficients are polynomially bounded, $\alpha(P, c)$ can be determined in polynomial time by solving the linear program over $\tilde{P}$ \cite{Khachiyan1980}.
    Furthermore, it clearly holds that $P_{\frac{1}{2}} \subseteq \tilde{P}$. 
    In the following we show that $\frac{1}{1 + \varepsilon} x^* \in P_{\frac{1}{2}}$ for any $x^* \in \tilde{P}$, 
    which concludes the proof.

    Assume there is a $\{0, \frac{1}{2}\}$-cut $a^T x \leq \lfloor a_0 \rfloor$ that is violated by $\frac{1}{1 + \varepsilon} x^*$.

    First, consider the case that the right-hand side $a_0$ is bounded by $k(\varepsilon)$. As mentioned earlier, this implies that at most $2k(\varepsilon)$ original constraints are involved in this cut. 
    However, this means that the cut is not violated by  $\frac{1}{1 + \varepsilon} x^* \in \tilde{P}$, a contradiction.

    Now suppose that $a_0 \geq k(\varepsilon)$. This implies that $\frac{1}{1 + \varepsilon} a^T x^* \leq \frac{1}{1 + \varepsilon} a_0 \leq \lfloor a_0 \rfloor$. The first inequality holds because $x^* \in \tilde{P} \subseteq P$, and the second
    inequality holds because $a_0 \geq k(\varepsilon) \geq 1 + \frac{1}{\varepsilon}$. This implies that $\frac{1}{1 + \varepsilon} a_0 \leq a_0 - 1 \leq \lfloor a_0 \rfloor$.\qed
\end{proof}

In the context of approximating the matching polytope, Sinha \cite{Sinha2018} showed that the construction above is essentially best possible.

Furthermore, our construction can be extended to the (mod $k$)-closure, which was introduced by Caprara, Fischetti, and Letchford \cite{Caprara2000}\footnote{An earlier conference version appeared in \cite{ConfCapraraModK}}.
For given $k$, it uses multipliers in $\{0, \frac{1}{k}, \ldots, \frac{k-1}{k}\}$.
\begin{definition}
    Given a polytope $P = \{x \in \mathbb{R}^n \colon Ax \leq b\}$, the (mod $k$)-closure $P_{\text{mod }k}$ is defined as
    \[
        P_{\text{mod }k} = \{x \in P \colon \lambda^T A x \leq \lfloor \lambda^T b \rfloor, \lambda \in \{0, \tfrac{1}{k}, \ldots, \tfrac{k-1}{k}\}^m, \lambda^T A \in \mathbb{Z}^n\}.
    \]
\end{definition}

In the proof of Theorem \ref{thm:PTAS}, the key element was the fact that we can find a lower bound for the contribution of an inequality to the right-hand side.
Since for the (mod $k$)-closure we have a bounded number of multipliers and a lower bound of $\frac{1}{k}$ for the multipliers, we obtain the following corollary.
\begin{corollary}
    Let $P = \{x \in \mathbb{R}^n \colon Ax \leq b\}$, with $A$ and $b$ integer. If $b > 0$, there is a PTAS for $P_{\text{mod }k}$.
\end{corollary}

If we assume $A \geq 0$, i.e., we are considering monotone (or packing-type) polytopes, we can w.l.o.g. assume that $b \geq \mathbbm{1}$. This holds because if $b_j = 0$ for some $j \in [m]$, we can project out all variables appearing in this constraint with non-zero coefficient. This shows that our construction also works for all monotone polytopes.
Moreover, when dealing with a fixed number $l$ of inequalities where the right-hand side equals zero, we can address this by increasing the value of $k(\varepsilon)$ by $l$. However, it remains unclear how to deal with a non-constant number of such inequalities.

%
\section{Conclusion}
We want to mention that, in general, the separation problem for $P_{\frac{1}{2}}$ is already hard if there are exactly three odd entries per column. This is implicit in a construction by Brugger and Schulz 
\cite{Brugger2022}. Furthermore, because the stable set problem is $\mathcal{NP}$-complete even when restricted to graphs whose degree is bounded by $3$ \cite{Garey1974}, 
their construction also implies that separation is hard even if all rows except one have at most three odd entries. 
A similar result (with four odd entries per column) can also be obtained by a reduction from Exact 3-Cover, using a construction similar to the one given in \cite{Caprara1996}.
 
Furthermore, it stays open how the primal separation procedures might be extended to the (mod $k$)-closure. In this setting, \textit{odd} has to be understood as not equivalent to $0$ mod $k$. 
In the construction above, the multipliers are implicitly determined by the set of involved inequalities, but when considering the (mod $k$)-closure, selecting the multipliers gives 
more freedom, which makes looking for a violated cut harder. However, maximally violated (mod $k$)-cuts, i.e., cuts that are 
violated by $\frac{k-1}{k}$ by the given fractional point, can be found in polynomial time, as was shown by Caprara, Fischetti, and Letchford \cite{Caprara2000}.

\subsubsection{Acknowledgements}The authors thank Matthias Brugger for interesting discussions and helpful comments and suggestions. This work was supported by the Deutsche Forschungsgemeinschaft (DFG, German Research Foundation) - Project 277991500.
%
%
%
\bibliographystyle{splncs04}
\bibliography{approx_and_primal_sep}
\end{document}